\documentclass[10pt,draftcls,onecolumn]{IEEEtran}    

\usepackage{amsfonts,amssymb,amstext,verbatim,amsmath,graphicx,color,enumerate}
\usepackage{amsfonts,amssymb,amstext,verbatim,amsmath,graphicx,color}
\usepackage{url}

\usepackage{algorithm}
\usepackage[noend]{algpseudocode}
\makeatletter
\newcommand{\StatexIndent}[1][3]{%
  \setlength\@tempdima{\algorithmicindent}%
  \Statex\hskip\dimexpr#1\@tempdima\relax}
\makeatother

\newtheorem{theorem}{\textbf{Theorem}}[section]
\newtheorem{algo}{\textbf{Dynamics}}

\newtheorem{propo}[theorem]{\textbf{Proposition}}
\newtheorem{definition}[theorem]{\textbf{Definition}}
\newtheorem{assumption}[theorem]{\textbf{Assumption}}

\newcommand{\rea}{\mathbb{R}}

\long\def\symbolfootnote[#1]#2{\begingroup%
\def\thefootnote{\fnsymbol{footnote}}\footnote[#1]{#2}\endgroup}

\newcommand{\map}[3]{#1: #2 \rightarrow #3}

\newcommand{\eps}{\varepsilon}

\newcommand{\real}{\mathbb{R}}

\newcommand{\realnonnegative}{\mathbb{R}_{\geq0}}

\newcommand{\integer}{\mathbb{Z}}
\newcommand{\integernonnegative}{\mathbb{Z}_{\ge 0}}

\usepackage{enumerate}

\makeatletter
\def\@opargbegintheorem#1#2#3{\trivlist
   \item[]{\bfseries #1\ #2\ (#3)} \itshape}
\makeatother

\IEEEoverridecommandlockouts                              

\begin{document}

\title{Stability of Open Multi-Agent Systems and Applications to Dynamic Consensus}%
\author{Mauro Franceschelli$^\ddagger$ and Paolo Frasca$^\dagger$
\thanks{\scriptsize $^\ddagger$M.~Franceschelli is with the Dept.\ of Electrical and Electronic Engineering, University of Cagliari, Italy, email: \ {\tt\scriptsize mauro.franceschelli@diee.unica.it}.}
\thanks{\scriptsize $^\dagger$P.~Frasca is with Univ.\ Grenoble Alpes, CNRS, Inria, Grenoble INP, GIPSA-lab, F-38000 Grenoble, France. email: \ {\tt\scriptsize paolo.frasca@gipsa-lab.fr}.}
\thanks{
\scriptsize $^\ddagger$ This work was supported by the Italian Ministry of Research and Education (MIUR) under call SIR 2014 with the grant ``CoNetDomeSys", code  RBSI14OF6H, and by Regione Autonoma della Sardegna with the 2016/2017 ``Visiting Professor Program". The research of P. Frasca was also partly supported by ANR (French National Science Fundation) via project HANDY, number ANR-18-CE40-0010.}
}

\maketitle

\begin{abstract}
In this technical note we consider a class of multi-agent network systems that we refer to as Open Multi-Agent Systems (OMAS): in these multi-agent systems, an indefinite number of agents may join or leave the network at any time. 
Focusing on discrete-time evolutions of scalar agents, we provide a novel theoretical framework to study the dynamical properties of OMAS: specifically, we propose a suitable notion of stability and derive sufficient conditions to ensure stability in this sense. These sufficient conditions regard the arrival/departure of an agent as a disturbance: consistently, they require the effect of arrivals/departures to be bounded (in a precise sense) and the OMAS to be contractive in the absence of arrivals/departures.  
In order to provide an example of application for this theory, we re-formulate the well-known Proportional Dynamic Consensus for Open Multi-Agent Systems and we characterize the stability properties of the resulting Open Proportional Dynamic Consensus algorithm.
\end{abstract}


\section{Introduction}\label{intro}


A multi-agent system is a dynamical model for the behavior of a possibly large group of agents, e.g., robots, devices, sensors, oscillators etc., whose pattern of interactions due to sensing, communication or physical coupling is modeled by a graph that represents the network structure of the system. Most literature on multi-agent systems considers networks of fixed size, i.e., number of agents, and then considers several kinds of scenarios such as time-varying network topologies. %
In this paper we explicitly consider a more radical scenario of {\em open} multi-agent system where {\em the set of agents is time-varying}, i.e., agents may join and leave the network at any time.

This situation is common in numerous applications, of which we mention a few. In the internet of things (IoT) and smart power grids, smart devices can join and leave the grid~\cite{FRANECC2016,Franceschelli20182541}; in social network (either online or offline~\cite{SG-PJ:12}) individuals can join or leave; in multi-vehicle coordination, the composition of a robotic team or platoon of vehicles can evolve with time.

Despite their ubiquity, open multi-agent systems have received surprisingly little attention in either control or in contiguous fields. Notions of ``open" systems can be found in the computer science literature \cite{Huynh2006,Pinyol2013} when referring to software agents and the problem of evaluating reputation in open environments, but not as dynamical systems. At times, dynamically evolving populations have also been considered in game theory, at least to show the robustness of price-of-anarchy results~\cite{shah10dynamics,TL-VS-ET:16}.
Despite the abundance of works in multi-agent systems from the systems and control community, openness is rarely explicitly included in a rigorous analysis, but rather explored by simulations as in~\cite{Zhu2010322}. In multi-robot systems, where adaptivity to addition/removal of robots is crucial, some architectures accommodate for dynamics teams but offer no performance guarantees~\cite{RP-PF-JWD-RC-FB:16}. Indeed, openness implies some conceptual difficulties in adapting control-theoretic notions such as state or stability. For this reason, some authors have recently proposed to circumvent some of the mathematical hurdles by embedding the time-varying agent set in a time-invariant finite superset~\cite{8364624}.
In a different perspective, others have aimed to describe the open multi-agent system through significant statistical properties:
insightful and encouraging results have been presented in \cite{7852357,hendrix2017}, where the authors study the problem of average-consensus by gossiping, and in~\cite{AbdHen_CDC17a}, where the authors study a max-consensus problem.

The contribution of this paper is twofold, as it covers both theoretical results and concrete examples. As a theoretical contribution, we introduce an abstract framework for discrete-time open multi-agent systems: this framework is based upon proper definitions of state evolutions, equilibria, and stability, and allows to establish useful stability criteria for a class of ``contractive'' open multi-agent systems. Instrumental to this development is extending the notion of (Euclidean) distance to apply to vectors that belong to different spaces and therefore have different length. This goal is achieved by our definition of {\em open distance} function. 
%
%

In order to provide a concrete example, we extend a distributed control protocol, namely Proportional Dynamic Consensus, to the open scenario, thereby defining the Open Proportional Dynamic Consensus algorithm. Its stability properties can be studied by our analysis tools. 

In the {\em dynamic consensus problem}, each of the nodes receives an input signal and is tasked to track the average of all inputs over the network. 
Our interest in dynamic consensus originates from its fundamental role in distributed control in general and specifically in the domain of smart grids. In the latter application, the object of the distributed estimation can be the time-varying average power consumption by the network. Thus, by considering the planned power consumption of each device as an external reference signal for each agent, a dynamic consensus algorithm can be used to estimate the time-varying average value of this potentially large set of reference signals. Since devices login and logout from the network without notice, the set of reference signals is, in general, time-varying.


The dynamic consensus problem has received significant attention, as demonstrated by the forthcoming tutorial~\cite{dyncons-csm18}. Since the early work in \cite{spanos2005dynamic}, a fundamental idea to render consensus protocols ``dynamic'' has been adding the derivative of each agents' own reference signal to a consensus filter that would thus track the time-varying average of the references. Several algorithms that exploit this mechanism have been proposed~\cite{Zhu2010322,Nosrati20122262,Kia2015112}: their main advantages are convergence speed and accuracy (which can be perfect for constant reference signals), while their main drawback is their lack of robustness with respect to errors in their initialization and, consequently, with respect to changes in the network composition. If the number of agents changes, these algorithms accumulate estimation errors that can severely deteriorate the estimation performance. 
Some algorithms, for instance those in \cite{Montijano20143131} and \cite{Freeman2006,Freeman2010,Freeman2015}, have instead shown superior robustness properties that can be useful to allow for the addition or removal of agents, 
even though their analysis has been so far limited to networks of fixed size. We take note that the main drawback of proportional dynamic consensus, as illustrated in~\cite{Freeman2006}, consists in a trade-off between steady-state error for constant reference signals, tracking error and convergence rate. In the recent work \cite{FRANCDC2016,FRANAUT2019}, the authors propose and characterize the so-called multi-stage dynamic consensus algorithm, which consists in a cascade of proportional consensus filters and has been proved to guarantee small steady-state and tracking error for a given convergence rate, thanks to exchanging a larger quantity of local information between the agents. Furthermore, the strategy proposed in \cite{FRANCDC2016,FRANAUT2019} has been shown to be implementable with asynchronous and randomized (gossip-based) local interactions.

In the recent conference publication~\cite{MF-PF:18}, we gave a preliminary account of our work that contains variations of some of the results presented here. Among the main differences, in the present paper stability is defined according to a normalized notion of distance, where the distance between two states is normalized by the square root of the number of agents; instead, this normalization was not used in~\cite{MF-PF:18}. Furthermore, the notion of contractive OMAS was not introduced in~\cite{MF-PF:18}. 

\paragraph*{Structure of the paper} In Section~\ref{problem statement} we introduce the framework of open multi-agent systems (OMAS) from a theoretical perspective and present an adaptation of two known distributed control protocols to this new framework. In Section~\ref{stability} we provide theoretical tools for the stability analysis of discrete-time OMAS and apply the results to two examples of distributed control protocols for OMAS. In Section~\ref{numericalexamples} we corroborate our results with numerical examples and finally, in Section~\ref{conclusions} we discuss some concluding remarks.

\section{Open dynamical systems} \label{problem statement}
For all time $k\in\integernonnegative$, let $G_k=\left(V_k,E_k\right)$ be a time-varying directed graph with time-varying set of agents (also called nodes) $V_k\subset\integer$ and time-varying set of edges $E_k\subseteq \left(V_k \times V_k\right)$. Set $V_k$ contains the labels corresponding to the agents that are active at time $k$.
The cardinality of set $V_k$, that is, the number of agents that belong to the network at time $k$, is denoted as $n_k=|V_k|$. To avoid trivialities, we assume that $n_k>0$ for all $k$. Two agents $v$ and $w$ are said to be neighbors at time $k$ if they share an edge at time $k$, i.e., $(v,w)\in E_k$. Let $N^v_k$ be the set of neighbors of node $v$ at time $k$, i.e., $N^v_k=\left\{w\in V_k: (v,w)\in E_k\right\}$. Let $\Delta^v_k=|N^v_k|$ denote the number of neighbors of agent $v$ at time $k$.

For each time $k$ and each agent $v\in V_k$, we associate a scalar {\em state} variable $x^v_k\in\real$ and an {\em input} variable $u^v_k\in \real$. Note that these variables are defined only at time instants such that $v\in V_k$. More generally, in this paper we shall call {\em open sequence} any sequence $\{y_k\}_k$ where $y_k$ is indexed over $V_k$.

With these ingredients we can define laws that describe how the open sequence $\left\{x_k\right\}_{k}$ evolves. However, we will not be able in general to write $x_{k+1}$ as a function solely of $x_k$: therefore, the evolution of $x_k$ does not constitute a ``closed'' dynamical system.
Instead, we shall take as given 
%
%
the open sequences $\left\{V_k\right\}_{k\in\integernonnegative}$ and $\left\{E_k\right\}_{k\in\integernonnegative}$, as well as the open sequence of inputs $\left\{u_k\right\}_{k\in\integernonnegative}$. Provided the consistency conditions that $E_k\subseteq \left(V_k \times V_k\right)$ and $u_k\in\real^{V_k}$ for all $k$, we shall define the evolution of the open sequence $\{x_k\}$ by a law
\begin{equation}\label{eq:open-general} 
x_{k+1}=f(x_k,V_k,V_{k+1},E_k, E_{k+1}, u_k, u_{k+1}).
\end{equation}
%
%
Such update rule should distinguish three kinds of nodes $v$, respectively belonging to the sets:

\begin{itemize}
\item $R_k=V_k\cap V_{k+1}$, i.e.,\ {\em remaining} nodes that belong to both $V_k$ and $V_{k+1}$;

\item $D_k=V_k\setminus V_{k+1}$, i.e.,\ {\em departing} nodes that belong to $V_k$ but not to $V_{k+1}$;

\item $A_k=V_{k+1}\setminus V_k$, i.e.,\ {\em arriving} nodes that belong to $V_{k+1}$ but not to $V_k$.
\end{itemize}
\smallskip
Since $x_k$ must take values in $\real^{V_k}$ for all $k$, the components corresponding to $D_k$ are simply left out from $x_{k+1}$. Instead, components in $A_k$ need to be ``initialized'' according to some rule. Finally, for all $v\in R_k$ there shall be a causal evolution law in the form
\begin{equation}\label{eq:inherent-dynamics}x^v_{k+1}=f^v(x_k,V_k,E_k,u_k).\end{equation}

For concreteness, we now describe an example of such map, which we call Open Proportional Dynamic Consensus (OPDC).
\begin{algo}[Open Proportional Dynamic Consensus (OPCD)] %
Let $\varepsilon>0$ and $\alpha\in (0,0.5)$.
At each time $k\in\integernonnegative$, each agent $v\in V_{k}$ measures a reference signal $u^v_k$ and updates its state $x^v_k$ as follows:
\begin{subequations}\label{opdc}
if $v\in R_{k}$, then 
\begin{align}
x^v_{k+1}=  x^v_{k} -\alpha (x^v_{k}-u^v_{k}) -\varepsilon \sum_{w\in N^v_{k}} (x^v_{k}-x^w_{k}); \label{part1}
\end{align}
if $v\in A_{k}$, then
\begin{align}
   x^v_{k+1}=
    \displaystyle u^v_{k+1}\label{part2}. 
   \end{align}
\end{subequations} 

\end{algo}

We observe that if $V_{k+1}= V_{k}$, i.e., the set of agents does not change, the OPDC reduces to what is called Proportional Dynamic Consensus. Namely, it reduces to the update~\eqref{part1}, which can be written in matrix form as
\begin{align}
x_{k+1}&=x_{k} -\alpha (x_{k}-u_{k})-\varepsilon L_{k} x_{k} \nonumber\\
&=\big((1-\alpha)I-\varepsilon L_{k}\big) x_k +\alpha u_{k} \nonumber\\
&=P_{k}x_k+\alpha u_{k} \label{globaldyn}
\end{align}
where matrix $P_{k}=(1-\alpha)I-\varepsilon L_{k}$.

We also observe that \eqref{part2} necessarily involves a component defined at time $k+1$, since $u_k$ is undefined when $v\in A_{k}$.

We will study OPDC under the assumption that the join/leave process guarantees some good behavior of sequence of graphs.
\begin{assumption}[Open Proportional Dynamic Consensus]\label{ass:about-opdc}
Consider the open dynamics~\eqref{opdc} and assume that for every $k\in\integernonnegative$,
\begin{enumerate}
\item Graph $G_k$ is undirected (that is, $(u,v)\in E_k$ if and only if $(v,u)\in E_k$);
\item $\max_{v \in V_k}\Delta_k^v \leq \frac{1}{2\varepsilon}$ for all $v\in V_k$;
\item the number of agents in the OMAS can not decrease too rapidly:
$\beta^2\leq \frac{|V_{k+1}|}{|V_{k}|}$ for some positive scalar $\beta$. 
\end{enumerate} 
\end{assumption}
We remark that Assumption~\ref{ass:about-opdc} can easily be satisfied in a distributed way and in particular does not require graph $G_k$ to be connected for any $k\in\integernonnegative$.
Let $\bar{\lambda}^{2}\ge0$ be a uniform lower bound to the algebraic connectivity $\lambda^{2,L}_{k}$ of the Laplacian matrix $L_k$ corresponding to graph $G_k$, that is, let $\bar{\lambda}^{2}\leq \lambda^{2,L}_{k}$. Such a constant always exists (since we allow it to be zero): when it is positive, it implies that all graphs are connected and that connectivity is uniformly good.


\section{Stability of open multi-agent systems}\label{stability}

%

In our general study of the stability of OMAS, we lie down our instruments in three steps: (i) we define suitable (sequences of) points that play the role of equilibria; (ii) we extend the notion of distance to operate on vectors of unequal length; (iii) we define a suitable notion of stability and give a sufficient condition for it.

\subsection{Points of interest and their stability}
We now define the concept of {\em trajectory of points of interest} which will be useful in the considered scenario of open multi-agent system.
\begin{definition}[\bf Trajectory of Points of Interest (TPI)]\label{def3.01}
Consider an open multi-agent system~\eqref{eq:open-general}. Assume that for every $k\geq 0$,
the equation
$$y=f(y,V_{k},V_{k},E_{k}, E_{k},u_{k}, u_{k})$$
has a unique solution and denote that solution as $x^{e}_{k}$. Then, the sequence $\left\{x^e_{k}\right\}_{k\in\integernonnegative}$ is called {\em trajectory of points of interest} of the open multi-agent system.  
\end{definition}

Observe that $x^{e}_{k_0}\in\real^{V_{k_0}}$ represents the hypothetical equilibrium of the dynamics followed by $x_{k}$ if the three sequences $V_{k}$, $E_{k}$ and $u_{k}$ would be kept constant for all $k\ge k_0$. Consequently, $x^e_{k_0}$ is determined only by information at time $k_0$: the time-variance of $V_{k}$, $E_{k}$ and $u_{k}$ does not imply any ambiguity in the definition of the sequence $x^e_{k}$.

The next definition~introduces a notion akin to a weak form of Lyapunov stability for open multi-agent systems.
\begin{definition}[\textbf{Open Stability of a Trajectory of Points of Interest}]\label{def2}
Let $x_{k}$ be the evolution of an open multi-agent system. A trajectory of points of interest $x^e(t)$ is said to be {\em open stable} if there exists a {\em stability radius} $R\ge 0$ with the following property: for every $\varepsilon>R$, there exists $\delta>0$ such that 
if $\frac1{\sqrt{n_0}}||x_0-x^e_0|| < \delta$, then $\frac1{\sqrt{n_k}}||x_k-x^e_k|| < \varepsilon$ for every $k\geq 0$. 
\end{definition}

Note that in this definition distances are normalized by the number of agents. Such normalization, which is trivial when the set of agents is invariant, is useful here because it allows for a fair comparison of distances evaluated in spaces of different dimension.


\subsection{Contractive OMAS}
Next we define a particular class of open-multi-agent systems of our interest.
\begin{definition}[Contractive OMAS] \label{contractiveOMAS}
Consider the open multi-agent system in \eqref{eq:open-general}. The OMAS is said to be \emph{contractive} if there exists 
$\gamma\in \left[0,1\right)$ such that for all $x,y \in \rea^{V_k}$ and  for all  $k\geq 0$
\begin{equation}\label{contract}
||f(x,V_k,V_k,E_k,u_k)-f(y,V_k,V_k,E_k,u_k)|| \leq \gamma ||x-y||. 
\end{equation}
\end{definition}

%
%


By Banach Fixed Point Theorem, every contractive OMAS has a TPI.
%
%
%
As an example, consider system~\eqref{opdc}. Under Assumption~\ref{ass:about-opdc}, the OPDC is a contractive OMAS and the solution $x^e_{k}$ is unique for every $k$ and can be computed as
\begin{align}\label{xeu} 
x^e_{k}=(I-P_{k})^{-1}\alpha u_{k}=\Big(I+\frac{\eps}{\alpha} L_{k}\Big)^{-1}u_{k}.%
\end{align}


\subsection{Open distance function}
Next, we define a so-called ``open" distance function which is used to evaluate the distance between two points with labeled components that belong to Euclidean spaces of different dimensions. In the particular case in which the two points have components with the same labels, i.e., the same agents, the open distance function reduces to the Euclidean distance.


\begin{definition}[Open distance function]\label{def3}
Let $V_1$ and $V_2$ be two finite sets of node indices. Let
$\map{d}{\real^{V_1}\times \real^{V_2}}{\realnonnegative}$ be defined as
\begin{equation}\label{opendistancefunction}
d(x,y)=  \displaystyle  \sqrt{\sum_{v\in V_1\cap V_2} (x^v-y^v)^2 + \sum_{v\in V_1 \setminus V_2} ({x^v})^2 + \sum_{v\in V_2 \setminus V_1} ({y^v})^2 }
\end{equation}
for any $x\in\real^{V_1}$ and $y\in\real^{V_2}$.
\end{definition}


%


Variants of Definition \ref{def3} can be given based on norms different from the $2$-norm.
The open distance~\eqref{opendistancefunction} satisfies several natural properties, which we summarize in the next statement.


\begin{propo}[\textbf{Properties of open distance functions}]\label{def3.1}
Function $d(x,y)$ in~\eqref{opendistancefunction} is such that for any vectors $x$, $y$, and $z$ of possibly different dimensions:
\begin{enumerate}
\item $d(x,y)\geq 0$;
\item $d(x,y)=d(y,x)$;
\item If $x=y$, then $d(x,y)=0$;
\item $d(x,z)\leq d(x,y)+d(y,z)$
\end{enumerate}
\end{propo}

\begin{proof}
Properties 1), 2), and 3) being evident, we now prove property 4), i.e., the triangle inequality.

Consider sets $V_x$, $V_y$, $V_z$ and 
define the union set $R= V_x \bigcup V_y \bigcup V_z$ and new vectors $\bar{x},\bar{y}, \bar{z} \in \rea^{R}$ where their generic component is defined as $\bar{x}^v=x^v$ if $i\in V_x$ and  $\bar{x}^v=0$ otherwise.
Since $\bar{x}$, $\bar{y}$, $\bar{z}$ belong to the same space $\rea^{R}$, it follows that the triangle inequality
$$d(\bar{x},\bar{y})\leq d(\bar{x},\bar{z})+d(\bar{z},\bar{y})$$ holds true since the open distance reduces to the ordinary Euclidean one.
The result follows because one can readily verify that $d(\bar{x},\bar{y})=d(x,y)$. 
\end{proof}

Note that the converse of the third implication (identity of indiscernibles) does not hold. Indeed, consider $x\in \real^{\{1,2\}}$ to be $x=[1, 0]$ and $y\in \real^{\{1\}}$ to be $[1]$. Then, $d(x,y)=0$ despite the two vectors being different.


Having this open distance available, we can naturally use it on open sequences to give the following definition.
\begin{definition}[Open sequence of bounded variation]\label{def0}
A sequence $\left\{y_{k}\right\}$ of points $y_k\in \real^{V_k}$ is said to have bounded variation if there exists a constant $B\ge 0$ such that $d(y_{k+1},y_{k}) \leq \sqrt{|V_{k+1}|}B$ for all $k\in\integernonnegative$.
\end{definition}
Note that this definition in fact normalizes the open distance by the number of components of the vectors.

An important special case is the trajectory of points of interest of an open multi-agent system.
\begin{definition}[Trajectory of points of interest (TPI) of bounded variation]\label{def1}
A TPI $\left\{x^e_{k}\right\}$ is said to have bounded variation if there exists $B$ such that $d(x^e_{k+1},x^e_{k}) \leq \sqrt{|V_{k+1}|}B$ for all $k\in\integernonnegative$.
\end{definition}


\subsection{Stability: sufficient conditions}


%
%
%



In order to provide a sufficient condition to ensure stability in the above sense, we will need to combine assumptions on both the associated TPI and on its join process.
The latter assumption will take the following form.

\begin{definition}[Bounded join process]\label{joinprocess}
A join process is said to be \emph{bounded} if each agent joins the OMAS with a state value such that 
$$\sqrt{\sum_{v\in A_{k}} \left(x^v_{k+1}-x^{e,v}_{k+1}\right)^2}\leq \sqrt{|V_{k+1}|}H \qquad \forall k\in\integernonnegative$$ for some $H\ge0$. 
\end{definition}


We are now ready to state our main stability result.

\begin{theorem}[Stability of Open Multi-Agent Systems] \label{theorem:stabilityOpenSystems3}
Consider an open multi-agent system as in \eqref{eq:open-general} with state trajectory $\left\{x_{k}\right\}$. Assume that
\begin{enumerate}
\item  the OMAS is contractive with parameter $\gamma\in [0,1)$;
\item  its TPI $\left\{x^e_{k}\right\}$ has bounded variation with constant $B$;
\item  the join process is bounded with constant $H$;
\item  $|V_{k+1}|\geq \beta^2 |V_{k}|$ for all $k$ with $\beta> \gamma$. 
\end{enumerate}
Then, the trajectory of points of interests is {\em open stable} (Definition \ref{def2}) with stability radius 
$$R=\frac{B+H}{1-\frac{\gamma}{\beta}}$$
\end{theorem}


\begin{proof}
At each iteration $k$ the agents first update their state, then some new agents may join and some may leave. By considering the open distance function, it holds
\begin{align}
d(x_{k+1},x^e_{k+1})=& \displaystyle \sqrt{\sum_{v\in V_{k+1} \cap V_{k}} (x^v_{k+1}-x^{e,v}_{k+1})^2} \nonumber\\
                       &\displaystyle \overline{+\sum_{v\in V_{k+1} \backslash V_{k}}  (x^v_{k+1}-x^{e,v}_{k+1})^2} \nonumber \\
                   \leq& \displaystyle \sqrt{\sum_{v\in V_{k+1} \cap V_{k}} (x^v_{k+1}-x^{e,v}_{k+1})^2} \nonumber\\
                       &+\displaystyle \sqrt{\sum_{v\in V_{k+1} \backslash V_{k}}  (x^v_{k+1}-x^{e,v}_{k+1})^2}\nonumber\\
                   \leq& \displaystyle \sqrt{\sum_{v\in V_{k+1} \cap V_{k}} (x^v_{k+1}-x^{e,v}_{k})^2} \nonumber\\
                       &+  \displaystyle\sqrt{ \sum_{v\in V_{k+1} \cap V_{k}} (x^{e,v}_{k+1}-x^{e,v}_{k})^2} \nonumber\\
                       &+ \displaystyle \sqrt{\sum_{v\in V_{k+1} \backslash V_{k}}  (x^v_{k+1}-x^{e,v}_{k+1})^2}. \label{eq:after-triangle}
\end{align}
%
Since the OMAS is contractive, we observe that
\begin{align} \label{eq:after-contraction}
\nonumber \sqrt{\sum_{v\in V_{k+1} \cap V_{k}} (x^v_{k+1}-x^{e,v}_{k})^2}\leq &\gamma \sqrt{\sum_{v\in V_{k+1} \cap V_{k}} (x^v_{k}-x^{e,v}_{k})^2} \\ \leq &\gamma d(x_k,x^e_k).
\end{align}
%
%
Now, note that the trajectory of points of interest is of bounded variation, implying 
\begin{equation} 
\displaystyle \sqrt{\sum_{v\in V_{k+1} \cap V_{k}} (x^{e,v}_{k+1}-x^{e,v}_{k})^2} \leq d(x^e_{k+1},x^e_{k})
\label{conditionB}
\leq \sqrt{|V_{k+1}|} B,
\end{equation}
and that the join process is bounded as per Definition \ref{joinprocess}, implying 
\begin{equation}\label{conditionH} 
\sqrt{\sum_{v\in V_{k+1} \backslash V_{k}}  (x^v_{k+1}-x^{e,v}_{k+1})}\leq \sqrt{|V_{k+1}|}H.
\end{equation}
Thus, by upper bounding the righthand side of \eqref{eq:after-triangle} by \eqref{eq:after-contraction}-\eqref{conditionB}-\eqref{conditionH}, we can write
\begin{align}
d(x_{k+1},x^e_{k+1}) \leq & \ \gamma d(x_{k},x^e_{k})  +\sqrt{|V_{k+1}|}B + \sqrt{|V_{k+1}|} H. \label{eq200} 
\end{align}
Let us now divide both sides of \eqref{eq200} by the square root of the cardinality of $|V_{k+1}|$
\begin{equation}\label{eq201} 
\frac{d(x_{k+1},x^e_{k+1})}{\sqrt{|V_{k+1}|}} \leq  \gamma \frac{d(x_{k},x^e_{k})}{\sqrt{|V_{k+1}|}} +B +H. 
\end{equation}
By assumption, $|V_{k+1}|\ge \beta^2 |V_{k}|$ where $\beta > \gamma$, thus we can write
\begin{equation}\label{eq202} 
\frac{d(x_{k+1},x^e_{k+1})}{\sqrt{|V_{k+1}|}} \leq  \frac{\gamma}{\beta} \frac{d(x_{k},x^e_{k})}{\sqrt{|V_{k}|}} +B +H. 
\end{equation}
%
%
This inequality implies that theTPI is open stable with stability radius 
%
$R=\displaystyle \frac{B+H}{1-\frac{\gamma}{\beta}}. $
%
%
\end{proof}

\section{Application: Open Proportional Dynamic Consensus}
We now characterize  the convergence properties of the Open Proportional Dynamic Consensus protocol.


\begin{theorem}[Stability of Open Proportional Dynamic Consensus] \label{OPCD-theorem}
Consider the Open Proportional Dynamic Consensus algorithm (OPDC) under Assumption \ref{ass:about-opdc} and assume that $\beta>1-\alpha$. Let  $\bar{u}_{k}=\frac{\mathbf{1}^Tu_{k}}{n}\mathbf{1}$, $\hat{u}_{k}=u_{k}-\bar{u}_{k}$. 
If the sequence of reference signals $\left\{u_k\right\}$ satisfies 
\begin{equation} \label{inputdisagreement}
\|\hat{u}_{k}\|_{\infty}\leq \Pi, \quad \Pi\geq 0
\end{equation}
and 
\begin{equation}\label{inputboundedvariation}
d(\bar{u}_{k+1},\bar{u}_k)\leq \sqrt{|V_{k+1}|}U,
\end{equation}
%
%
%
%
 then the OPDC is open stable with stability radius 
$$R= \frac{\left(1+
\frac{2}{1+\frac{\varepsilon}{\alpha} \bar{\lambda}^{2}} + \frac{1}{\beta} \frac{1}{1+\frac{\varepsilon}{\alpha} \bar{\lambda}^{2}}\right) \Pi+U}
{1-\frac{1-\alpha}{\beta}}$$
%
%
%
\end{theorem}


\begin{proof}
The proof is divided into four steps which lead to the application of Theorem~\ref{theorem:stabilityOpenSystems3}.

\textbf{Step 1}: 
As we have already observed right before \eqref{xeu}, system~\eqref{opdc} under Assumption~\ref{ass:about-opdc} is a contractive OMAS with $\gamma=1-\alpha$.

\textbf{Step 2:} If the sequence of reference signals $u_k$ \eqref{inputboundedvariation} satisfies \eqref{inputdisagreement} and \eqref{inputboundedvariation}, then the TPI is of bounded variation with constant $$B=\frac{1}{1+\frac{\varepsilon}{\alpha} \bar{\lambda}^{2}}\left(1+\frac{1}{\beta}\right) \Pi+U.$$

\smallskip

We start the proof of step 2 by exploiting the triangle inequality property of the open distance function
\begin{equation}\label{eq:use-triangle} 
d(x^e_{k+1},x^e_{k})\leq d(x^e_{k+1},\bar{u}_{k+1})+ d(x^e_{k},\bar{u}_{k}) +d(\bar{u}_{k+1},\bar{u}_{k}).
\end{equation}
The points of interest 
are 
\begin{align}
    x^{e}_k = (I-P_k)^{-1} \alpha u_k= (\alpha I + \varepsilon L_k)^{-1}\alpha (\bar{u}_k + \hat{u}_k).
\end{align}
Since $(\alpha I + \varepsilon L_k)^{-1}\bar{u}_k=\bar{u}_k$ for any $L_k$ we can write 

\begin{align}
    x^{e}_k-\bar{u}_k = \alpha (\alpha I + \varepsilon L_k)^{-1}\hat{u}_k.
\end{align} 
Now, since the eigenvector corresponding to the largest eigenvalue of  $(\alpha I + \varepsilon L_k)^{-1}$ is $\mathbf{1}$ and $\mathbf{1}^T\hat{u}_k=0$, it holds  
\begin{align}\label{eq:barra-e-tratto}
   \| x^{e}_k-\bar{u}_k\|_2 = \|\alpha (\alpha I + \varepsilon L_k)^{-1}\hat{u}_k\|_2\leq \frac{1}{1+\frac{\varepsilon \lambda^{2,L}_{k}}{\alpha}}\|\hat{u}_k\|_2,
\end{align} 
where $\left(1+\frac{\varepsilon \lambda^{2,L}_{k}}{\alpha}\right)^{-1}$ is the second largest eigenvalue of matrix $\alpha \left(\alpha I + \varepsilon L_k\right)^{-1}$.

Then, the distance between the point of interests and the reference signal at time $k$ satisfies 
\begin{align}
d(x^e_{k},\bar{u}_{k})& = \|x^e_{k}-\bar{u}_{k}\|_2 \leq \frac{\alpha}{\alpha+\varepsilon \lambda^{2,L}_{k}} \|\hat{u}_{k}\|_2 \nonumber \\
 & \leq \frac{1}{1+\frac{\varepsilon}{\alpha} \lambda^{2,L}_{k}} \sqrt{|V_k|}\|\hat{u}_{k}\|_{\infty}.
\end{align}
By noting that $\frac{|V_{k+1}|}{\beta^2}\geq |V_k|$, $\|\hat{u}_{k}\|_{\infty}\leq \Pi$ and $d(\bar{u}_{k+1},\bar{u}_{k})\leq \sqrt{|V_{k+1}|} U$ for some $U\geq 0$, 
it follows from \eqref{eq:use-triangle} that
\begin{align}
d(x^e_{k+1},x^e_{k}) &\leq  \sqrt{|V_{k+1}|}\left(  \frac{1}{1+\frac{\varepsilon}{\alpha} \bar{\lambda}^{2}}\left(1+\frac{1}{\beta}\right) \Pi+U\right) \nonumber \\
                     & =  \sqrt{|V_{k+1}|}B. 
\end{align}




\textbf{Step 3:} The join process of the OPCD is bounded according to Definition \ref{joinprocess} with  $$H=\left(1+ \frac{1}{1+\frac{\varepsilon \bar{\lambda}^{2}}{\alpha}}\right)\Pi.$$ 

\smallskip

In the OPDC algorithm new agents join with a state value equal to their reference signal. Since from \eqref{eq:barra-e-tratto} at time $k+1$,

$$\|x^e_{k+1}-\bar{u}_{k+1}\|_2 \leq \frac{\alpha}{\alpha+\varepsilon \lambda^{2,L}_{k}} \|\hat{u}_{k+1}\|_2,$$
and $|u^v_{k+1}-\bar{u}_{k+1}|\leq \Pi$ by assumption, by recalling that  $A_{k+1}=V_{k+1}\backslash V_{k}$, it holds 
\begin{align*} 
\sqrt{\sum_{v\in A_{k+1}} \left(x^v_{k+1}-x^{e,v}_{k+1}\right)^2} & \leq \sqrt{|V_{k+1}\backslash V_{k}|} (1+ \frac{\alpha}{\alpha+\varepsilon \lambda^{2,L}_{k+1}})\Pi \nonumber\\
                                                                & \leq \sqrt{|V_{k+1}|} (1+ \frac{1}{1+\frac{\varepsilon \bar{\lambda}^{2}}{\alpha}})\Pi 
\end{align*}
Thus, the join process of the OPCD algorithm is bounded according to Definition \ref{def2} with $$H=\left(1+ \frac{1}{1+\frac{\varepsilon \bar{\lambda}^{2}}{\alpha}}\right)\Pi.$$



\textbf{Step 4}: By Theorem~\ref{theorem:stabilityOpenSystems3}, the TPI of the OPCD algorithm is open stable with stability radius
$R= \frac{B+H}{1-\frac{\gamma}{\beta}}.$ 
\end{proof}

When $\bar \lambda^{2}=0$, that is, the join process does not guarantee a uniform connectivity, then the stability radius in Theorem~\ref{OPCD-theorem} takes the simpler form 
$$R_0=\frac{(3 \beta +1) \Pi+ \beta U}{\alpha+\beta-1}$$ and we observe that in general $R\le R_0$.

\section{Numerical examples}\label{numericalexamples}

In this section we show a numerical example of the OPCD algorithm. Our simulations are performed as follows. We considered as tuning parameters $\varepsilon=0.01$, $\alpha=0.1$. The simulated scenario consists of a network of $200$ agents at the initial time, with initial values choosen uniformly at random in the interval $\left[-5000,5000\right]$. The initial graph is generated as an Erd\H{o}s-R\'enyi graph with edge probability $p=0.05$. At each iteration one agent leaves with probability $0.06$ and one agent joins with probability $0.1$: each arriving agent creates random edges with probability $0.05$ with all other agents.  Input reference signals are constant and, when agents join the network, are chosen uniformly at random in the interval $\left[0,1\right]$. 

After describing our simulation setup, we present one typical realization.
To begin with, in Figure \ref{figureOPCD1} we show the evolution of the number of agents and in Figure \ref{figureOPCD3} we show the evolution of the normalized open distance $\frac{d(x^e_k, \bar{u}\mathbf{1})}{\sqrt{|V_k|}}$, that is, the distance between the current point of interest and the average of the input reference signals given to the agents. The value of this quantity depends on the OPCD parameters, in particular it could be reduced by decreasing the parameter $\alpha$. 

We then proceed to exemplify the stability properties of the OPDC. To this purpose, Figure \ref{figureOPCD2} shows the evolution of the normalized open distance $\frac{d(x_k,x^e_k)}{\sqrt{|V_k|}}$, that is, the distance of the state $x_k$ of the network from the current point of interest $x^e_k$. It can be observed that this normalized distance remains bounded after a transient decrease.

This phenomenon is consistent with the stability analysis given in Theorem \ref{OPCD-theorem}. Even though our analysis makes deterministic assumptions and therefore does not in principle allow drawing conclusions on this randomized evolution, we can {\em a posteriori} verify that the simulated join/leave process has satisfied the assumptions of Theorem~\ref{OPCD-theorem} with minimum algebraic connectivity $\bar{\lambda}^2=0.9037$,  $|V_{k+1}|\geq \beta^2|V_k|$ with $\beta=0.9975$, largest degree equal to $20$, $\Pi=0.5139$, and $U=0.0001785$.
Therefore, the result implies a stability radius equal to $R=17.375$, which appears to be a conservative estimate according to Figure~\ref{figureOPCD2}.  

Moreover, in Figure \ref{figureOPCD4} we show the evolution of the normalized open distance $\frac{d(x_k, \bar{u}\mathbf{1})}{\sqrt{|V_k|}}$, which represents the distance between the network state and the average of the input reference signals. Estimating the latter quantity is the objective of the OPCD protocol. This estimation error can be seen to converge to a bounded value despite the open nature of the multi-agent system.

For a useful comparison, in Figures \ref{figureOPCD5} and \ref{figureOPCD6} we show a simulation with of the PDC algorithm with a {\em fixed set of agents} ($n=200$) and constant reference signals. In Figure \ref{figureOPCD6} it can be seen that the network state converges to its equilibrium point (up to machine precision), in contrast with the finite error in Figure \ref{figureOPCD2}. In Figure \ref{figureOPCD5} it can be seen that the network state converges to a steady-state which has a bounded error with respect to the average reference signals: in comparison with Figure \ref{figureOPCD4}, the Open PDC reaches a similar steady-state error (albeit at slower pace) as its classical PDC counterpart.

%



\begin{figure}[t!]
\centering
\includegraphics[width=0.5\textwidth]{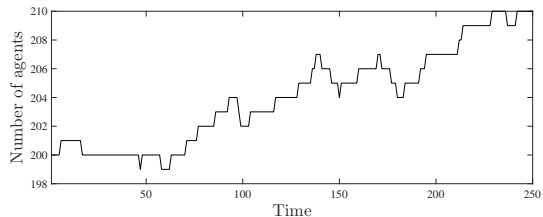}
\caption{Time-varying number of agents.}\label{figureOPCD1}
\end{figure}

\begin{figure}[t!]
\centering
\includegraphics[width=0.5\textwidth]{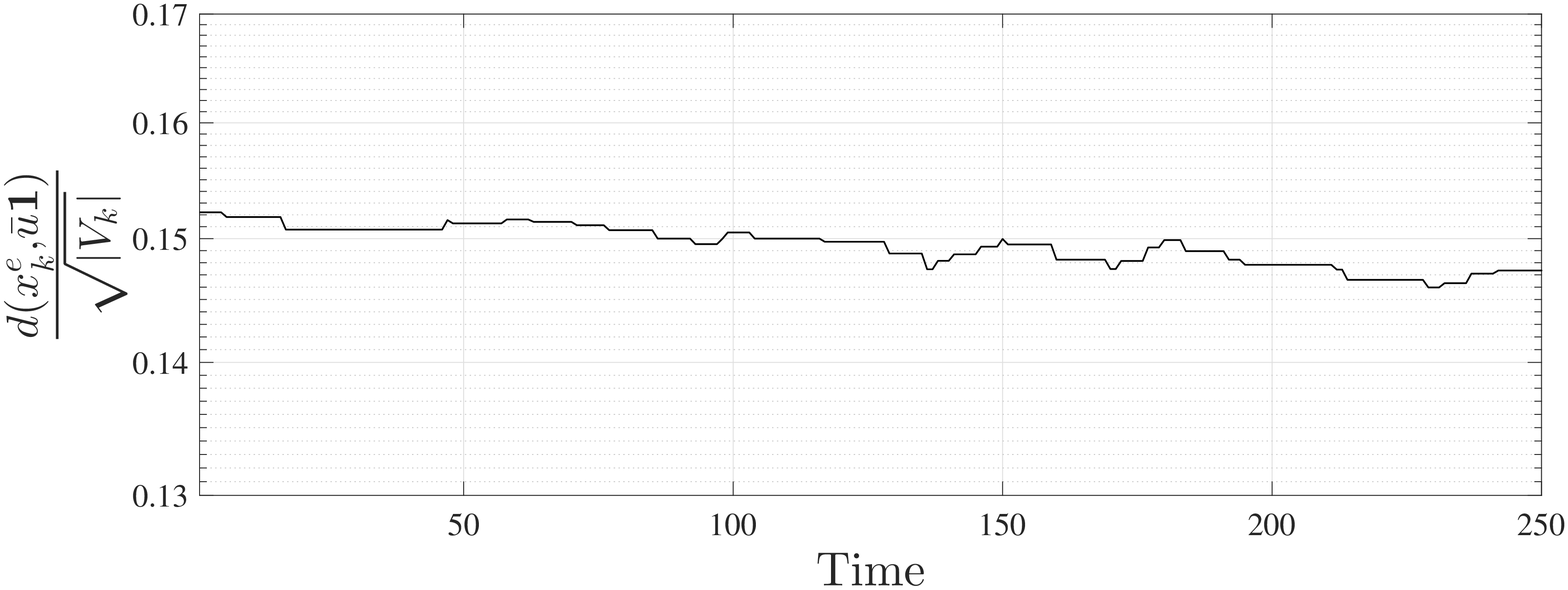}
\caption{Evolution of the normalized open distance between average reference input and point of interest: $\frac{d(x^e_k, \bar{u}\mathbf{1})}{\sqrt{|V_k|}}$.} \label{figureOPCD3}
\end{figure}

\begin{figure}[t!]
\centering
\includegraphics[width=0.5\textwidth]{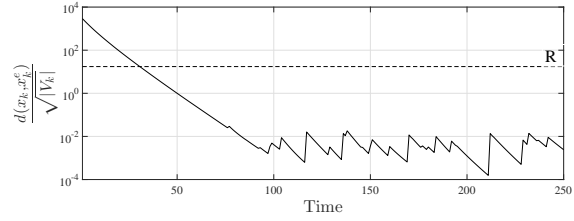}  
\caption{Evolution of the normalized open distance between network state and point of interest: $\frac{d(x_k,x^e_k)}{\sqrt{|V_k|}}$. } \label{figureOPCD2}
\end{figure}

\begin{figure}[t!]
\centering
\includegraphics[width=0.5\textwidth]{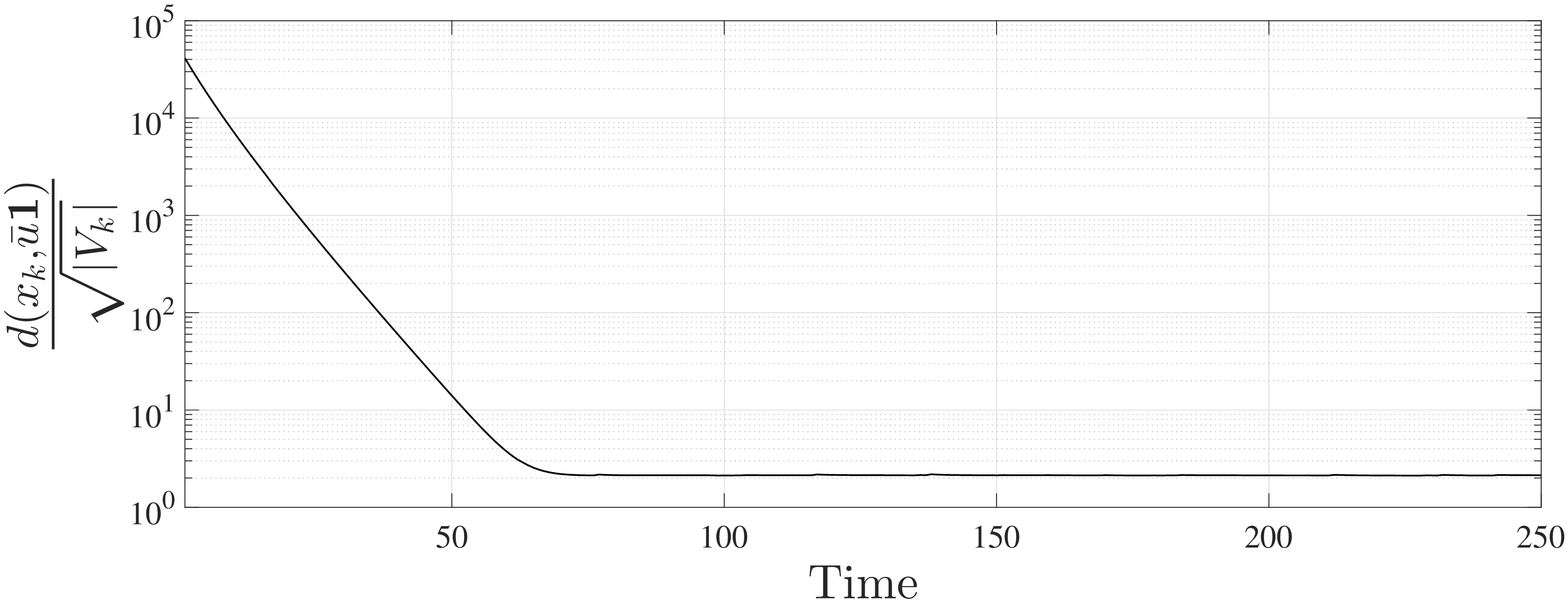} 
\caption{Evolution of the normalized open distance between network state and average reference input: $\frac{d(x_k,\bar{u}\mathbf{1}))}{\sqrt{|V_k|}}$.} \label{figureOPCD4} 
\end{figure}

\begin{figure}[t!]
\centering
\includegraphics[width=0.5\textwidth]{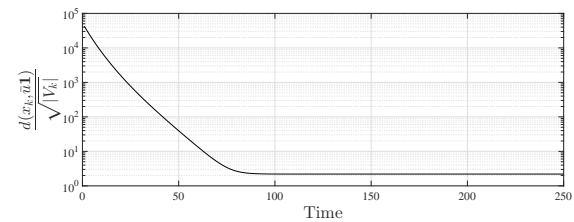} 
\caption{Evolution of the normalized open distance between network state and average reference input: $\frac{d(x_k,\bar{u}\mathbf{1}))}{\sqrt{|V_k|}}$ in the case of time-invariant number of agents $n=200$.} \label{figureOPCD5} 
\end{figure}

\begin{figure}[t!]
\centering
\includegraphics[width=0.5\textwidth]{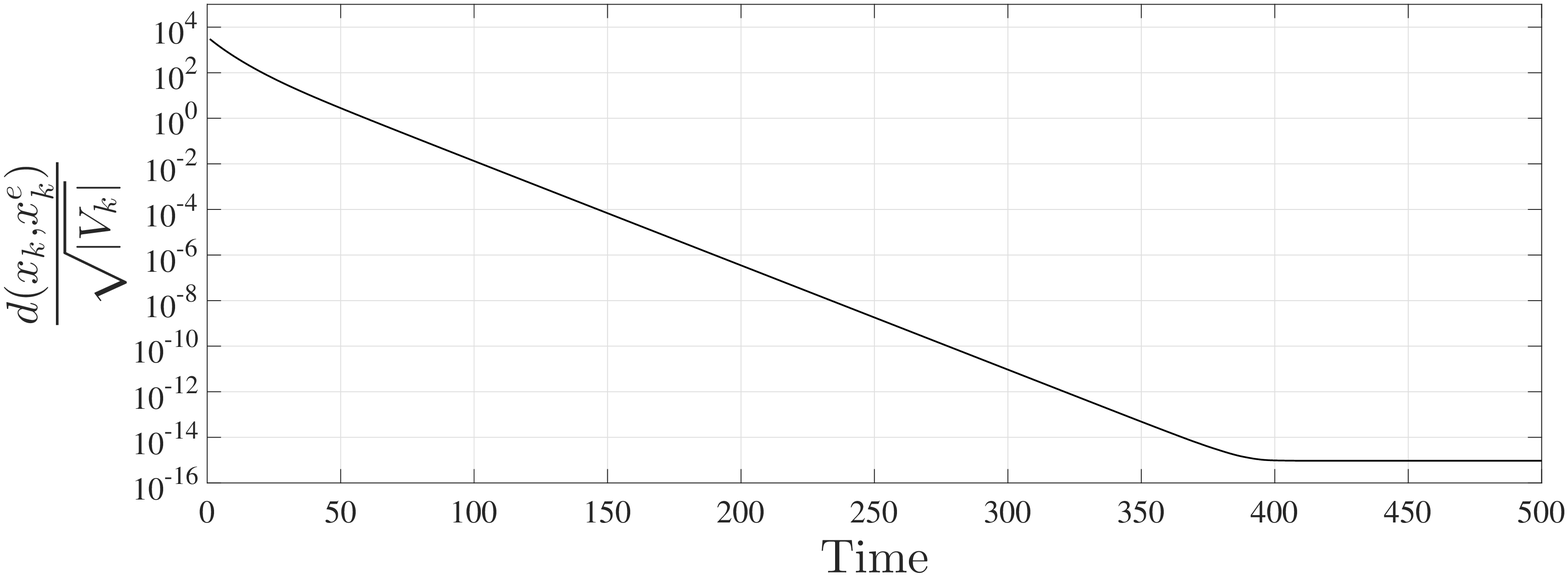} 
\caption{Evolution of the normalized open distance between network state and point of interest: $\frac{d(x_k,x^e_k)}{\sqrt{|V_k|}}$  in the case of time-invariant number of agents $n=200$.} \label{figureOPCD6}
\end{figure}
\section{Conclusions}\label{conclusions}
In this paper we proposed a theoretical framework for stability analysis of discrete-time open multi-agent systems. Standard system-theoretic tolls do not apply directly to OMAS, because of the evolution of their state space. For this reason, we had to propose several new definitions, including suitable definitions of state evolution and of stability. The proposed notion of stability has two features: (1) the distance from the origin is normalized by the number of agents; and (2) the definition disregards what happens within a certain distance from the origin (we refer to this distance as stability radius). In order to study the evolution and the stability of OMAS, it is necessary to compare states that belong to different spaces. To this purpose, we defined the open distance function and used it to establish criteria for stability in the proposed open scenario. In particular, we showed that multi-agent systems whose dynamics (up to arrivals and departures of agents) can be defined by contraction maps are stable according to our definition and their stability radius depends upon the properties of the join and leave mechanisms in the network. Furthermore, we applied our results to an adaptation to OMAS of the proportional dynamic consensus protocol. Future work should pursue two complementary direction: building up a more general and comprehensive theory, while at the same time investigate other classes of open-multi agent systems and propose novel ``open'' distributed coordination algorithms.

\bibliographystyle{IEEEtran}
\bibliography{biblio}

\end{document}